\documentclass[conference]{IEEEtran}
\IEEEoverridecommandlockouts
\usepackage{cite}
\usepackage{amsmath,amssymb,amsfonts}
\usepackage{algorithmic}

\usepackage{graphicx, caption, subcaption}
\usepackage{amsthm}
\newtheorem{definition}{Definition}

\usepackage{wrapfig}
\usepackage{pifont}
\newcommand{\cmark}{\ding{51}}%
\newcommand{\xmark}{\ding{55}}%

\usepackage{textcomp}
\usepackage{multirow}

\usepackage[hyphens]{url}

\usepackage[hyphenbreaks]{breakurl}

\usepackage{xcolor}
\def\BibTeX{{\rm B\kern-.05em{\sc i\kern-.025em b}\kern-.08em
    T\kern-.1667em\lower.7ex\hbox{E}\kern-.125emX}}
\newtheorem{theorem}{Theorem}
\usepackage[linesnumbered,ruled,vlined]{algorithm2e}

\SetCommentSty{mycommfont}

\SetKwInput{KwInput}{Input}                
\SetKwInput{KwOutput}{Output}              

\def\BibTeX{{\rm B\kern-.05em{\sc i\kern-.025em b}\kern-.08em
    T\kern-.1667em\lower.7ex\hbox{E}\kern-.125emX}}
\begin{document}

\title{Reward Mechanism for Blockchains Using Evolutionary Game Theory\\
}

\author{\IEEEauthorblockN{Shashank Motepalli,
Hans-Arno Jacobsen}
\IEEEauthorblockA{Department of Electrical and Computer Engineering,
University of Toronto\\
shashank.motepalli@mail.utoronto.ca,
jacobsen@eecg.toronto.edu}}

\IEEEoverridecommandlockouts
\IEEEpubid{\begin{minipage}{\textwidth}\ \\[12pt]
 \IEEEauthorrefmark{1}This research has in part been funded by NSERC.\\[4pt]
 978-1-6654-3924-4/21/\$31.00 ©2021 IEEE
\end{minipage}}
\maketitle
\IEEEpubidadjcol

\begin{abstract}
Blockchains have witnessed widespread adoption in the past decade in
various fields. The growing demand makes their scalability and
sustainability challenges more evident than ever. As a result, more
and more blockchains have begun to adopt proof-of-stake (PoS)
consensus protocols to address those challenges. One of the
fundamental characteristics of any blockchain technology is its
crypto-economics and incentives. Lately, each PoS blockchain has
designed a unique reward mechanism, yet, many of them are prone to
free-rider and nothing-at-stake problems. To better understand the ad-hoc design of reward mechanisms, in this paper, we develop a reward
mechanism framework that could apply to many PoS blockchains. We
formulate the block validation game wherein the rewards are
distributed for validating the blocks correctly. Using evolutionary
game theory, we analyze how the participants' behaviour could
potentially evolve with the reward mechanism. Also, penalties are
found to play a central role in maintaining the integrity of
blockchains.
\end{abstract}

\section{Introduction}
\label{introduction}
Bitcoin began the blockchain revolution by enabling a peer-to-peer
version of electronic cash~\cite{nakamoto2019bitcoin}. Over the past
decade, Bitcoin grew as a digital reserve worth more than a trillion
dollars~\cite{pound_2021}. Ethereum further extended these concepts to
create programmable money using smart
contracts~\cite{wood2014ethereum}. This idea led to numerous
applications~\cite{adams2021uniswap,
  mejeris2021blockchain}. However, the increasing popularity has
exposed the drawbacks of such systems, that is proof-of-work-based
(PoW) systems do not scale~\cite{mentzer2018impact,
  kokoris2018omniledger}. Furthermore, PoW raises concerns about its
high carbon footprint~\cite{stoll2019carbon}. Lately, more and more
blockchains are migrating to proof-of-stake (PoS) protocols for
state-replication while addressing scalability and sustainability
challenges~\cite{kiayias2017ouroboros, kwon2014tendermint,
  gilad2017algorand}. PoS systems employ \textit{validators} for
processing the transactions, similar to the miners in PoW systems. In
the PoS system, the participants often lock their tokens to become
validators.
%
These validators participate in the consensus process to decide on the
next state of the ledger. With each new block added, i.e., with every
state transition, tokens are minted to reward the validators for
processing the transactions. A validator can also increase their stake
by buying the tokens traded in the secondary markets. Since the entire
security in PoS systems relies on stake, the economics of tokens are
more critical than ever.


\begin{table*}[]
\caption{Reward mechanism in prominent PoS blockchains}
\centering
\label{tab:various-bc}
\begin{tabular}{|l|c|c|c|c|c|}

\hline
\textbf{}             & \multicolumn{1}{l|}{\textbf{Reward every validator}} & \multicolumn{1}{l|}{\textbf{Proportional to stake}} & \textbf{Penalize not voting} & \multicolumn{1}{l|}{\textbf{Penalize conflicting transactions}} &\multicolumn{1}{l|}{\textbf{References}} \\ \hline\hline
\textbf{Algorand}     & \cmark                                                    & \cmark                                                        & \xmark                     & \xmark    &   ~\cite{gilad2017algorand,algorand-faq, fooladgar2020incentive}                                                                 \\ \hline
\textbf{Avalance}     & \xmark                                                    & \xmark                                                         & \xmark                     & \xmark        &~\cite{aval-token}                                                                 \\ \hline
\textbf{Cardano}      & \cmark                                                    & \xmark   & \xmark                     & \xmark     &\cite{david2018ouroboros, kiayias2017ouroboros}                      \\ \hline
\textbf{Cosmos}       & \xmark  & \xmark     & \cmark                    & \cmark    &~\cite{cosmos, kwon2014tendermint}                                                                    \\ \hline
\textbf{Ethereum 2.0} & \xmark           & \cmark          & \cmark    & \cmark   &~\cite{eth2stake, eth2pos}                                                                     \\ \hline
\textbf{Polkadot}     & \cmark & \cmark & \xmark    & \cmark      &~\cite{polkadot}              \\ \hline
\end{tabular}
  \vspace{-2em}
\end{table*}


Each PoS blockchain tends to approach its \textit{reward mechanism}
design in a unique way (see Table~\ref{tab:various-bc}). We designedly
establish it as \textit{reward mechanism} and not an incentive
mechanism because we define a reward as both, the sum of
%
%
\textit{incentive} and \textit{penalty}. Some PoS blockchains such as
Polkadot and Cardano reward every registered validator irrespective of
whether they contributed to the consensus
process~\cite{david2018ouroboros, polkadot}. Algorand goes even
further to reward everyone who owns its tokens~\cite{algorand-faq,
  fooladgar2020incentive}. Whereas, Avalanche, Cosmos, and Ethereum
2.0 only reward validators who participate in the consensus
process~\cite{aval-token, cosmos, eth2stake}. In case everyone is
rewarded equally, a validator can become a \textit{free rider}
enjoying the rewards without doing the work. Furthermore, blockchains
differ a lot in terms of the reward they give. Most blockchains
provide rewards proportional to the stake validators' pledge. On the
other hand, Ethereum 2.0 and Polkadot reward equally irrespective of
the validators' stake~\cite{eth2stake, polkadot}. As there is
\textit{nothing at stake} to lose in some designs, validators can act
in ways that might affect the integrity of the ledger. Few blockchains
introduced penalties to address the \textit{nothing at stake}
problem. However, there are stark differences among PoS systems
regarding penalties. Ethereum 2.0 and Cosmos penalize (i.e., slash)
even if the validator is not participating in the consensus
protocol~\cite{cosmos, eth2stake}. Polkadot penalizes if and only if
the validator signs conflicting transactions~\cite{polkadot}. Whereas,
Algorand, Avalance, and Cardano foresee no notions of penalty in their
design.

%
%
The problem this paper addresses is understanding the reward mechanism
in PoS blockchains by proposing a unified framework. This problem
is interesting because there are a plethora of contradicting
approaches adopted by current blockchains. Since the reward mechanism
is crucial to the integrity of PoS systems, it is critical to
understand the implications of reward mechanisms. The central idea we
analyze is how to design rewards (i.e., payoffs) such that rational
behaviour would promote honesty. To abstract, there are three questions
to consider towards reward mechanism design in a PoS system. Firstly,
\textit{should all validators be rewarded equally, even if they don't contribute to consensus?} Secondly, \textit{should penalities be imposed for integrity in blockchains?}  Thirdly, \textit{would honesty emerge as a stable state?}

In this paper, we use game theory to formulate block validation as a
game. Block validation is a necessary action for progressing the
blockchain. The validators reach a consensus by validating and voting
on each proposed block and get rewarded for their goodwill. It is a
multiplayer game. The rewards for validators are not just dependent on
individual strategy but also on what the majority votes. Assuming
validators are rational, i.e., selfish, and want to maximize their
payoff, our goal is to design rewards for the game such that voting on
the honest strategy always remains rational. We begin with a simple
reward matrix, adjust it to address free-riding and nothing-at-stake
problems.
%
Finally, we use evolutionary game theory to analyze long term impact of whether the honest
strategy is a stable state. We ran simulations to confirm the same.

%
%

The contributions of this paper are three-fold. (1) To the best of our
knowledge, we are among the first to formulate the block validation
game that can be applied on any PoS blockchain. (2) We are also among
the first to analyze how validators change strategies in blockchains
using evolutionary game theory. (3) We establish the need for
penalties for the integrity of PoS blockchains.



This paper is organized as follows. In Section~\ref{background}, we
discuss necessary background about PoS and evolutionary game
theory. We review related work in Section~\ref{related-work}. In
Section~\ref{block-validation-game}, we formalize the block validation
game addressing the free-rider and nothing-at-stake problems. We
derive the evolutionarily stable strategies (ESS) in
Section~\ref{evolutionariy-stable-strategy} and experimentally confirm
ESS in Section~\ref{simulations}. Finally, we conclude with future
directions in Section~\ref{conclusion}.

\section{Background}
\label{background}

\subsection{Proof of Stake}
%
%
Proof of Stake is a state-replication mechanism designed for public
blockchains wherein parties that might not trust the other need to
reach consensus. The core of PoS relies on the simple idea that one
would not devalue the assets one owns. Hence, the stakeholders accept
the responsibility of maintaining the security of the PoS
blockchain. In a PoS system, the stakeholders can stake their tokens
to become validators on the network. The validators in PoS assume the
role of miners in PoW systems by verifying the transactions and
creating new blocks.

PoS systems are still nascent and evolving. However, there are two
broad categories of PoS concerning participation from validators:
\begin{itemize}
  \item Probabilistic PoS: This is similar to Bitcoin's Nakamoto
    consensus.
    %
    Instead of deciding the block proposer proportional to
    computational power in Nakamoto consensus, a block proposer is
    selected based on its stake in the PoS system. The higher its
    stake, the higher is its probability of being selected in the
    random block comit process. After being selected, the block
    proposer proposes the next block to be added on top of the
    previous block.  Cardano~\cite{david2018ouroboros} adopts this
    approach.
  \item BFT-based PoS: This is a weighted Byzantine fault-tolerant
    (BFT) consensus protocol wherein weights are proportional to the
    stake validators own. Each block requires a byzantine agreement
    for it to be appended on the blockchain. When the majority of
    validators vote, we reach finality on the given block.
    %
    In BFT-based PoS systems, more validators need to participate in
    the consensus process. Examples of this implementation include
    Algorand~\cite{gilad2017algorand},
    Cosmos~\cite{kwon2014tendermint}.
\end{itemize}

There are many variants of PoS systems, and the above is not an
exhaustive list. For instance, Ethereum 2.0 is adopting a hybrid of
both probabilistic PoS and BFT-based PoS~\cite{eth2pos}. In this
paper, we concentrate on BFT-based PoS protocols.

Let us walk through a lifecycle of a generic transaction on a
BFT-based PoS blockchain. A client generates a transaction and sends
it to a validator who gossips the transaction to other validators.
The validators verify and store valid transactions in their memory
pool. When a validator is selected to become block proposer, it goes
through its memory pool of valid transactions and accumulates a few of
them to generate a block to broadcast to all other validators. The
validators simulate the transactions in a block, verify their
correctness, and sign their approval or disapproval for that block~\cite{chacko2021my}.
%
%
The block proposer collects signatures from all the validators.
%
If a quorum confirms the validity of a block, we reach finality on the
block. Since most BFT protocols tolerate up to one-third of its nodes
to fail, the quorum size is set to two-thirds~\cite{
  kwon2014tendermint, zhang2021prosecutor}.

\subsection{Evolutionary Game Theory}
The origin of evolutionary game theory dates back to 1973 when John
Maynard Smith and George R. Price formalized it to study the evolving
populations for lifeforms in biology in their work, "The logic of
animal conflict"~\cite{smith1973logic}. Though it started as a concept
in evolutionary biology to explain the Darwinian evolution of species
proving the survival of the fittest~\cite{vincent2005evolutionary},
the theories of evolutionary game theory were adopted by
economists~\cite{friedman1998economic, witt2016specific,
  mailath1998people}, psychologists~\cite{buss2015evolutionary,
  laland2011sense, tooby2005conceptual}, among others. In classical
game theory, the success of a strategy depends on the strategy
itself. In contrast, in evolutionary game theory, the game is played
multiple times among numerous players. The success of an individual
strategy depends not only on the strategy itself but also on the
frequency distribution of alternative strategies and their successes.

The fitness of a strategy in a population is the payoff received for
choosing a strategy, given the population state, i.e., the frequency
of each strategy~\cite{easley2010networks}. In terms of the usual
conventions of a game, the higher the payoff, the higher the
fitness. Analogous to the Nash equilibrium in classical game theory,
there is the \textit{evolutionarily stable state} (ESS) in
evolutionary game theory. A strategy is in ESS if it is unaffected by
a small fraction of invading mutants that choose a different
strategy~\cite{smith1976evolution, cowden2012game}. In other words, if
for a sufficiently small mutant population, the fitness of the
incumbent strategy is still higher than that of any of the strategies
by mutants, the players stick to the incumbent strategy, making it an
ESS. We use ESS to prove that our reward design for the blockchain
network is stable.

\section{Related Work}
\label{related-work}
This work is closely related to two categories of related work. Firstly, the economics of PoS blockchains. Secondly, the applications of game theory and evolutionary game theory in blockchains. We reviewed the economics of existing PoS blockchains in Table \ref{tab:various-bc}.
%
%

%
%
%

%
%
A comprehensive survey summarizing numerous applications of game
theory to blockchains is provided by Liu et
al.~\cite{liu2019survey}. There is large interest in the
game-theoretic analysis of attacks such as selfish mining
attacks~\cite{zhang2020analysing}, block withholding
attacks~\cite{kroll2013economics}, among others. Few recent works
formalize evolutionary games for selection of mining pools and shards
in PoW-based blockchains~\cite{liu2018evolutionary,
  ni2019evolutionary, kim2019mining}. However, these works focus on
PoW consensus protocols and have limited applicability to PoS systems.

Our work also relates to reward design using game theory, studied for
example in~\cite{iyer2018crypto, chang2020incentive,
  chiu2019incentive}. However, their focus remained on PoW single-shot
games and not evolutionary ones. A data-sharing incentive model for
smart contracts based on evolutionary game theory is proposed
in~\cite{xuan2020incentive}. Nevertheless, the reward design at the
protocol layer is not analyzed. We are among the first to study
reward mechanism in blockchains using evolutionary game theory.

\section{The Block Validation Game}
\label{block-validation-game}
We employ game theory to model the strategic interactions among the
validators in a blockchain network. A normal form game $G$($N$,
$M_{i}$, $U_{i}$) has three key components: i) the players $N$, ii)
the strategies, denoted by $M$, iii) the payoffs for all the potential
strategies, denoted by $U$. We use payoff and reward interchangeably
in this work. In our game, the validators $v$, in the validator set,
verifying the blocks are the players, stated below as $N$.
\begin{equation}
N = \{v_i \mid v_i \in Validator Set\} 
\end{equation}

%
These validators contribute to the consensus protocol by approving or
disapproving the blocks proposed to be appended to the ledger. Each
block $B_i$ contains a set of transactions $tx_k$ as follows.
\begin{equation}
B_i=\{tx_{0}, tx_{1}, tx_{2}, ... tx_{n}\}
\end{equation}

Every validator owns a local copy of the blockchain state generated
from listening to blocks on the network.
%
The validators simulate the transactions on their local blockchain
state to verify their correctness. The correctness of a transaction
depends on several factors including, but not limited to, the
authentication of the sender's signature, the balance on the sender's
account, verifying the sequence number of the accounts of the sender,
and the logic of the smart contract. We define a valid block as
follows.

\begin{definition}
  Valid Block. The block $B_i$ is valid if and only if every
  transaction in the block is correct.
\end{definition}
If $correct(tx)$ provides the correctness of a transaction, a valid
block is represented as follows.
\begin{equation}
{valid(B_i) = correct(tx_0) ...\land correct(tx_n), \forall tx_k\,\exists B_i} 
\end{equation}

The validators have two potential strategies. They can either act
honestly or maliciously. An honest validator approves valid blocks and
disapproves invalid blocks. Any actions that deviate from honest
behaviour are considered malicious, including disapproving a valid
block or approving a block that is not valid. If a validator does not
participate in the consensus process, it is considered malicious
behaviour. Malicious behaviour accounts for all types of failures in the
system, such as crash failures, network failures, and Byzantine
failures. The validators' strategies, where $M$, the set of
strategies, is defined as follows. Here, $h$ and $m$ represent honest
and malicious strategies, respectively.
\begin{equation}
M= \{h, m\}
\end{equation}

%
%
By signing with their private key, validators respond to the proposer,
stating whether they approve or reject the proposed block. A validator
can at the most influence but cannot control the decisions of
others. The validators could only pursue pure strategies, i.e., a validator can either
act honestly or maliciously on a single block but not a combination of both. 
%
An individual validator has little or no influence on the final
decision, as the final decision depends on the population state. The
population state $X$ for block $B_i$ is the distribution of the
strategies of the validator set in that block; it is given below.
\begin{equation}
   X_{B_i} = \begin{bmatrix}
    x_{h}\\
    x_{m} \\
\end{bmatrix} 
\end{equation}
where $x_{h}$ $\geq$ 0, $x_{m}$ $\geq$ 0 and $x_{h}$ + $x_{m}$ = 1. In
the following, we look into the reward matrix for this game.

\subsection{Universal Reward: The Basic Reward Matrix}

The rewards play a crucial role in directing validators'
strategies. We assume all the validators are rational, i.e., they are
selfish and play to maximize their rewards. Our goal is to design
rewards such that the best response for a rational validator would be
acting honestly. A validator's reward is not just dependent on its
strategy. The rewards also depends on how other validators act, i.e.,
the population state. Each validator decides individually, knowing
very little about the others' strategies. BFT consensus protocols
require a quorum $\mathbb{Q}$, at least two-thirds of the validator
set $N$, to reach a consensus. The quorum is a subset of the validator
set as follows.
\begin{equation}
\mathbb{Q}_{B_m} \subset N, \; 
|\mathbb{Q}_{B_m}| \geq \frac{2}{3} \times |N|
\end{equation}
We assume that every round terminates with a decision on a block,
either by reaching an honest quorum or a malicious quorum. In other
words, we assume to reach a two-thirds quorum for each block. If $x
\geq 0.67$, then $X_{B_k}$ and $X'_{B_k}$ are the population states,
for honest and malicious quora, respectively.
\begin{equation}
   X_{B_k} = \begin{bmatrix}
    x\\
    1-x \\
\end{bmatrix},   X_{B_k}' = \begin{bmatrix}
    1-x\\
    x \\
\end{bmatrix} x \geq 0.67
\end{equation}
In case we do not achieve a consensus, the system proposes a nil block
after a timeout. A nil block is an empty block, leading to no approved
transactions and no rewards to anyone in that consensus round. Though
BFT protocols assume an honest majority, we do not want to rule out
the possibility of a malicious majority, given the high stakes
involved.  Malicious behaviour could be subtle to be labelled by the
system.For instance, if a few validators form a cartel and choose to censor
transactions by certain entities; this behaviour is not malicious
according to BFT consensus. Even in the cases of hard forks to adopt
new rules is subjective and does not constitute malicious
behaviour. In this paper, we limit our discussion about what
constitutes a Byzantine transaction in the specified network and
assume an altruistic quorum as a basis for designing rewards in the
system.  In other words, we reward anyone who adheres with a
two-thirds majority, and we consider anyone who digresses from it as
Byzantine. Every block provides rewards to all the validators for
contributing to the consensus.

%
The security and integrity of the blockchain network depends on the liveness of
validators and the availability of data. The validators incur an
expense \textit{e} in terms of network, computation, and storage
resources for verifying blocks. The blockchain system provides
incentives to reward honest behaviour and meet desired security. The
incentive $i$ is an aggregate sum of transaction fees on each
transaction $tx$ of that block and the block rewards $R$. The users
pay transaction fees for processing their transactions. The blockchain
system provides the block rewards for generating new blocks. The
effective reward $r$ for reaching a consensus on the next block $B_m$
for a validator $v_i$ is the difference between incentive $i$ and cost
$c$.  Here, $N$ is the size of the validator set. Since PoS and BFT
protocols are not computationally expensive operations, we assume the
costs \textit{e} to be lower, when compared to PoW mining. The
incentive $i$ is given as follows.
\begin{equation}
i_{B_m} = \bigg(\sum_{0}^{N}{fees(tx_i)} + R\bigg)_{B_m}
\end{equation}
The reward is the difference between incentive and expense for running
the node, it is stated below.
\begin{equation}
r_{B_m}({v_i})= \bigg(\frac{i}{N}\bigg)_{B_m}-  \textit{e}
\end{equation}
\begin{theorem} 
The reward has to be positive for the network to be secured.
\end{theorem}
\begin{proof}
If the expense incurred is negative, it means that running the
validator incurs an expense without any incentive. Since validators
are rational, they would choose to maximize their reward by not
participating in the consensus process. If the majority of honest
validators act rational, the security of the system is
compromised. For the system to be secured, the incentive should be
greater than the expense. In other words, rewards should be positive,
$r_{B_m} \geq 0$.
\end{proof}

Some validators might exhibit malicious behaviour such as
double-spending to gain benefits, often much higher than the
rewards. Let the benefit gained by the malicious cartel of validators
be $B$. It is fair to assume that some malicious validators might be
more beneficial than others. Let the reward for Byzantine behaviour be
represented by $b$, where $0\leq b \leq B$.

Each validator selects their strategy for every block. The reward for
the validator is dependent on the strategy chosen by the quorum (see
Table~\ref{tab:payoff}). The reward matrix $U$ is stated in
Equation~\ref{U1}. If the validator acts honest in an honest quorum,
it earns an effective reward of $r$. If the validator acts malicious
in an honest quorum, it reaps a reward of $r+\textit{e'}$, where $0
\leq \textit{e'} \leq \textit{e}$ could be the cost saved by not
running the validator. In the case of a malicious quorum, an honest
and malicious validator would be making $r$ and $r+b$,
respectively. In case of malicious behaviour in the malicious quorum, the reward could be $r+b+\textit{e'}$. Because we assume the Byzantine reward to be generally higher than the cost, $b \geq \textit{e'}$, we considered it as $r+b$ in the reward matrix.

\begin{table}[]
\centering
\caption{Universal Reward Case: The Reward Matrix}
\label{tab:payoff}
\begin{tabular}{c|l|l|l|}
\cline{2-4}
\multicolumn{1}{l|}{}                            & \multicolumn{3}{c|}{Quorum}    \\ \hline
\multicolumn{1}{|c|}{\multirow{3}{*}{Validator}} &           & honest & malicious \\ \cline{2-4} 
\multicolumn{1}{|c|}{}                           & honest    & $r$      & $r$         \\ \cline{2-4} 
\multicolumn{1}{|c|}{}                           & malicious & $r+$\textit{e'}      & $r + b$     \\ \hline
\end{tabular}
  \vspace{-2em}
\end{table}

\begin{equation}
\label{U1}
  U = \begin{bmatrix}
    $r$       & $r$ \\
    $r+ \textit{e'}$       & $r+b$\\
\end{bmatrix}
\end{equation}

\textbf{Analysis.} A validator receives a maximum reward when it leads a
malicious cartel. Acting
maliciously always has better incentives than acting honestly, so a rational
validator would choose  the no-regret strategy to act maliciously. This could lead to a
malicious quorum, hence, a security concern. A rational validator
chooses the sure-win case of maximizing their reward by \textit{e} by
not participating in consensus, as they do not incur computational
expenses.

\subsection{Reward For Work: Addressing Free-Riding}
As seen in the previous section, malicious behaviour is an optimal
strategy for a rational validator. This behaviour results in cartels or
worse, never leads to reaching a quorum. The validators would be
rewarded more ($r+$\textit{e}) even without validating as long as they
are in the validator set. This behaviour is a classical case of the
free-rider problem wherein the validators earn "something for
nothing"~\cite{li2017novel}. If more honest validators act rationally,
the security of the blockchain network is compromised. To overcome
this challenge, only the validators who participate in consensus
should earn incentives. In other words, the validators who do not
participate in the consensus process should not earn incentives. The
reward function is updated to incentivize only those who are in the
quorum as follows.
\begin{equation}
r_{B_m}({v_i})=\begin{cases}
          (\frac{i}{N`})_{B_m} - \textit{e} \quad &\text{if} \, v_i \in \mathbb{Q} \\
          0 \quad &\text{if} \, v_i \notin \mathbb{Q} \\
     \end{cases}
\end{equation}
Here, $N`$ is the size of the quorum. It ranges from two-thirds to
cardinality of the validator set. Let $e`$ be the cost for malicious
behaviour, which is equal to $e$ if it participates and 0, otherwise.

The reward matrix is updated to reward only those who worked (see
Table~\ref{tab:payoff2}). With an honest quorum, a validator choosing
an honest strategy would earn an effective reward $r$, whereas a
validator choosing a malicious strategy would result in a maximum of
zero payoffs if it saves on computational resources. However, if you
act honestly in the malicious quorum, you incur an expense \textit{e}
without any reward. A malicious validator that forms the malicious
quorum earns the reward of $r+b$. The updated reward matrix is given
below.
\begin{equation}
  U = \begin{bmatrix}
    $r$       & -\textit{e} \\
    -\textit{e}`       & $r+b$\\
\end{bmatrix}
\end{equation}

\begin{table}[]
\centering
\caption{Reward For Work Case: The Reward Matrix addressing Free Rider Problem}
\label{tab:payoff2}
\begin{tabular}{c|l|l|l|}
\cline{2-4}
\multicolumn{1}{l|}{}                            & \multicolumn{3}{c|}{Quorum}    \\ \hline
\multicolumn{1}{|c|}{\multirow{3}{*}{Validator}} &           & honest & malicious \\ \cline{2-4} 
\multicolumn{1}{|c|}{}                           & honest    & $r$      & -\textit{e}        \\ \cline{2-4} 
\multicolumn{1}{|c|}{}                           & malicious & -\textit{e}`      & $r + b$     \\ \hline
\end{tabular}
  \vspace{-2em}
\end{table}

\textbf{Analysis.} The updated reward matrix ensures that a rational
validator would participate in the consensus process. They continue to
act honestly, assuming an honest quorum. However, if the malicious
majority takes over the quorum, an honest validator would not earn
incentives while incurring operational expenses. The best response for
a rational validator would be signing all the blocks irrespective of
whether they are valid or not.

\subsection{Penalty Case: Addressing Nothing at Stake}
Rational validators being selfish agents to maximize their reward,
approve all the blocks irrespective of whether the blocks are valid or
not. This behaviour guarantees to reward them for every block that gets
appended onto the blockchain. Though this is not a serious threat if
only a few validators do it. However, if malicious players build a
quorum, this would be a security threat to the system. It affects the
social welfare of the system. It could easily lead to a situation
wherein we might have conflicting forks of the ledger, leading to
multiple ledger states.

Consider a scenario where a malicious validator double spends to
create multiple versions of the truth. If the malicious validator is
selected as a block proposer, as shown in Fig.~\ref{nothing}, it
could propose two different blocks from the same parent block. The
malicious validator could double-spend by broadcasting blocks that
have conflicting transactions. As discussed, the rational validators
would approve both the blocks to be in the quorum to increase their
reward. If we reach a quorum of rational validators, we would have
both blocks reaching quorum. Both these chains could grow indefinitely
in parallel. The fork resolution strategies such as the longest chain
rule~\cite{nakamoto2019bitcoin} or GHOST
protocol~\cite{sompolinsky2015secure} cannot resolve which fork is the valid fork, leading to multiple sources of truth. 
Since the equilibrium for a validator is to act rationally, this is a
security threat. We need to re-design the reward mechanism to avoid
this practice.

\begin{figure}[htbp]
\centerline{\resizebox{\width}{4.75cm}{\includegraphics[scale=0.4]{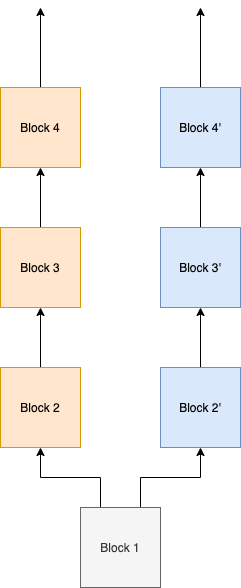}}}
\caption{Forked chains, where attacker double spends after Block~1}
\label{nothing}
\vspace{-1em}
\end{figure}

Our goal is to ensure that rational behaviour is acting honestly. We
need to punish the validators that deviate from the honest
strategy. One option could be not paying the future incentives for
validators that diverge from acting honestly for the next $n$ blocks,
where $n$ is a fixed number. It could also be jailing the validator from
the validator set. Another option is imposing a penalty for diverging
from the honest quorum.

We borrow from Behavioural Economics to improve the integrity of the
system. Kahneman and Tversky concluded that "losses loom larger than
gains" while explaining the concepts of Loss aversion in Prospect
Theory~\cite{tversky1979prospect},~\cite{tversky1991loss}. They
conducted numerous experiments to prove that people when presented
with alternatives, go for choices that either lead to sure wins or
avoid losses. Theory of Loss Aversion concluded that the pain of
losing has a much stronger psychological impact than the pleasure of
gaining the same amount. Few studies have further proved that
participants are also willing to behave dishonestly to avoid a loss
than to make a gain~\cite{pfattheicher2016misperceiving}. The loss
aversion experiments explain why penalties are more effective than
incentives in motivating people to behave in a certain
way~\cite{gachter2009experimental}. We use the same principles to
motivate validators to participate honestly instead of being malicious
validators in the network.

To address nothing at stake, we penalize validators that deviate from
the honest quorum. Let $p$ be the penalty for deviating from the
quorum, where $p>0$. The stake lost in penalty is designed to be much
greater than the incentive received for being in the quorum. The loss
to the stake also reduces the chances of becoming a validator in the
future. The updated reward function for a validator $v_i$ is given
below.
\begin{equation}
r_{B_m}({v_i})=\begin{cases}
          (\frac{i}{N`})_{B_m} - \textit{e} \quad &\text{if} \, v_i \in \mathbb{Q} \\
          - $p$ \quad &\text{if} \, v_i \notin \mathbb{Q} \\
     \end{cases}
\end{equation}

The updated matrix is given below (see
Table~\ref{tab:payoff3}). Acting honestly in an honest quorum gives
reward $r$ while acting maliciously in a honest quorum leads to a
penalty of $\textit{e}` + p$, since $p >> \textit{e}$, we ignore
\textit{e}. In the rare event of malicious majority, acting honestly
would result in penalty $p$. If a validator is acting maliciously in a
malicious majority, it earns a reward $r+b$ depending on its role in
the malicious cartel.

\begin{equation}
  U = \begin{bmatrix}
    $r$       & -p \\
    -p       & $r+b$\\
\end{bmatrix}
\end{equation}

\begin{table}[]
\centering
\caption{Penalty Case: The Reward Matrix addressing Nothing at Stake}
\label{tab:payoff3}
\begin{tabular}{c|l|l|l|}
\cline{2-4}
\multicolumn{1}{l|}{}                            & \multicolumn{3}{c|}{Quorum}    \\ \hline
\multicolumn{1}{|c|}{\multirow{3}{*}{Validator}} &           & honest & malicious \\ \cline{2-4} 
\multicolumn{1}{|c|}{}                           & honest    & $r$      & -p        \\ \cline{2-4} 
\multicolumn{1}{|c|}{}                           & malicious & - p      & $r + b$     \\ \hline
\end{tabular}
  \vspace{-2em}
\end{table}

\textbf{Analysis.} Given the reward matrix, the validators are better
off by acting honestly in an honest quorum than acting maliciously in
a malicious quorum.
%
%
However, under the core assumption of any blockchain network of the
majortity of the network being honest, the best response for a
rational validator is to act honestly.

\section{Evolutionarily Stable Strategy}
\label{evolutionariy-stable-strategy}
In the previous section, we designed the payoffs using classical game
theory that studies one-shot games where the players make rational
choices evaluating probable outcomes. Their outcomes were not just
dependent on their strategy but also on the population state of the
network. Under the assumption of an honest quorum in the penalty case,
we proved that a rational validator's best response would be to act
honestly under the no-regret strategy. Unlike classical one-shot
games
%
%
, the same set of validators play the game multiple times, once for
each block consensus round, heading towards a potential shift in their
strategy in the following rounds. Malicious validators could form
cartels to persuade the rational validators, who are acting honestly,
to alter their strategy over the next few blocks. To study how the
population state is evolving as the blockchain progresses and confirm
if acting honestly is a stable state, we apply evolutionary game
theory to our setting. We study evolutionarily stable strategies
(ESS), a strategy that cannot be invaded by other strategies.

%
%
We can determine ESS by a simple experiment. Assume all the validators
choose a particular strategy. In other words, the whole population of
validators is either honest or malicious. If a small number of mutants
deviate from the incumbent strategy, we analyze whether these minority
mutants have a better or worse payoff than that of the incumbent
strategy. If they do have a better payoff, the incumbent validators
would eventually shift to a mutant strategy. If the mutant strategy
performs worse, no mutants would invade the incumbent strategy making
the incumbent strategy is an ESS~\cite{smith1982evolution}.
%

\subsubsection{Everyone is honest}
Consider the case where all validators are honest ($h$), then the
incumbent strategy is given as follows:
\begin{equation}
   X_{B_h} = \begin{bmatrix}
    1\\
    0 \\
\end{bmatrix} 
\end{equation}

Let us now examine whether validators being honest is an ESS. When all
the validators are honest, if a small percentage of a mutant
population $\epsilon$ invade the network, the incumbent validators
continue remaining honest, we consider the honest strategy as ESS. We
can tolerate up to one-third of the total validators as being mutants,
acting maliciously. The population state with mutants is $X'_{B_h}$.
\begin{equation}
   X'_{B_h} = \begin{bmatrix}
    1- \epsilon\\
    \epsilon \\
\end{bmatrix} \epsilon \leq 0.33 
\end{equation}

The fitness $\mathcal{F}$ is the payoff for choosing a particular
strategy, given the population state $X'_{B_h}$. The fitness of the
incumbent strategy, \textit{honest}, is $\mathcal{F}(h)$, and the fitness of
the mutant strategy, \textit{malicious}, is $\mathcal{F}(m)$. We analyze the fitness for the three cases:
%
%
\begin{itemize}
\item Universal reward case: $\mathcal{F}(h) = r$ and $\mathcal{F}(m)
  = r$, $\mathcal{F}(h) = \mathcal{F}(m)$ and $\mathcal{F}(m) > 0$.
\item Reward for work case: $\mathcal{F}(h) = r$ and $\mathcal{F}(m) =
  -e$, $\mathcal{F}(h) > \mathcal{F}(m)$ and $\mathcal{F}(m) = 0$
\item Penalty case : $\mathcal{F}(h) = r$ and $\mathcal{F}(m) = -p$,
  $\mathcal{F}(h) > \mathcal{F}(m)$ and $\mathcal{F}(m) < 0$
\end{itemize}

\begin{figure*}
    \begin{subfigure}[t]{0.3\textwidth}
      \includegraphics[width=\linewidth]{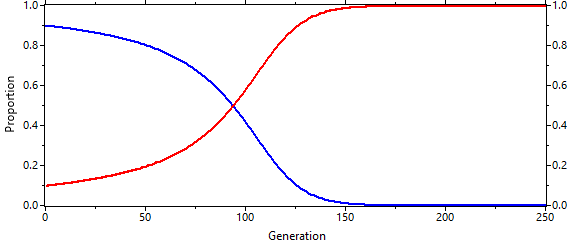}
      \caption{Universal reward case, the malicious minority takes over}\label{fig:case1-9010}
    \end{subfigure}
    \hfill
    \begin{subfigure}[t]{0.3\textwidth}
        \includegraphics[width=\linewidth]{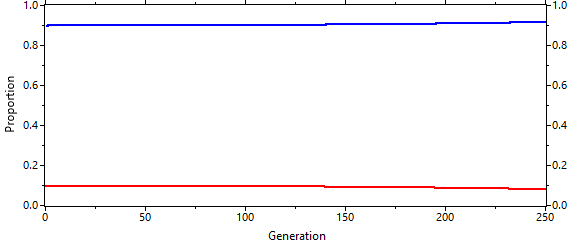}
        \caption{Without penalties, honest strategy can tolerate 10\% malicious}\label{fig:case2-9010}
    \end{subfigure}
        \hfill
    \begin{subfigure}[t]{0.3\textwidth}
        \includegraphics[width=\linewidth]{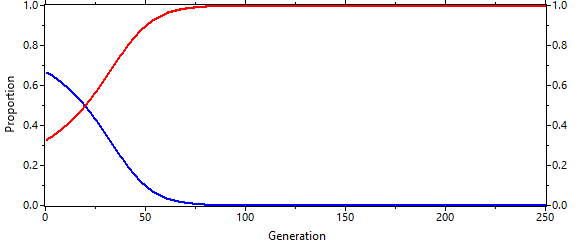}
        \caption{Without penalties, the malicious minority takes over, starting with one-third}\label{fig:case2-simulation}
    \end{subfigure}
    
     \begin{subfigure}[t]{0.3\textwidth}
        \includegraphics[width=\linewidth]{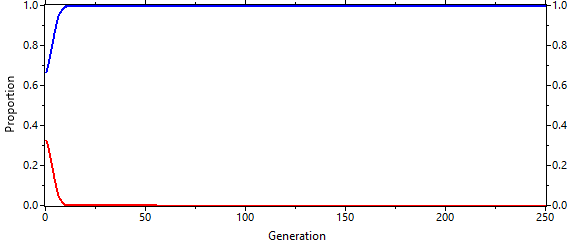}
        \caption{If penalty is 50\% of their total stake}\label{fig:c3500p}
    \end{subfigure}
    \hfill
    \begin{subfigure}[t]{0.3\textwidth}
        \includegraphics[width=\linewidth]{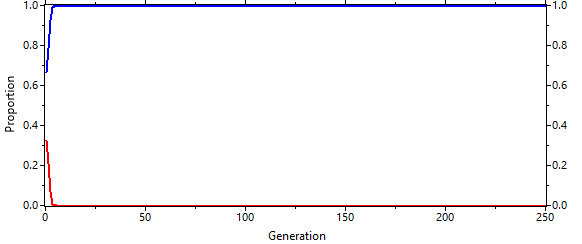}
        \caption{If penalty is 100\% of their total stake}\label{fig:c31000p}
    \end{subfigure}
        \hfill
    \begin{subfigure}[t]{0.3\textwidth}
        \includegraphics[width=\linewidth]{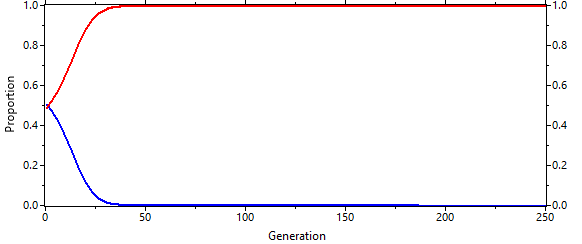}
        \caption{51\% honest and 49\% malicious would lead to malicous ESS}\label{fig:c3_5149}
    \end{subfigure}

    \caption{Experiments to study evolutionarily stable states}
    \vspace{-1em}
\end{figure*}

If the fitness of the incumbent's strategy is greater than that of the
mutant's strategy, the incumbents won't alter their strategy. This
makes the incumbent's strategy an ESS. In the universal reward case,
we have a mixed ESS since both the strategies have equal fitness; this
situation could lead to a potential shift over time. In the other
cases, being honest is ESS. Though the incumbent's honest strategy is
ESS in both cases. Here, the penalty case has a strong ESS, because
according to the theory of loss aversion, penalizing is more powerful
in influencing the decision than not gaining incentives.

\subsubsection{Everyone is malicious}

We consider the case where all the validators choose being malicious
as their incumbent strategy. Similar to the previous situation of an
honest population, we can prove that everyone being malicious would
remain malicious and the malicious strategy is an ESS.

%
%
In the universal reward case, we have no pure ESS. Both the strategies
of being honest and malicious can be invaded by others. For the rest
of the cases, we have two ESS, either honest or malicious populations,
subject to the relative values of the incentives, benefits and
penalties.

\begin{theorem} 
The security of a PoS blockchain system depends on the population
state during the genesis.
\end{theorem}
\begin{proof}
Let us consider the penalty case. With decent penalties, both honest
and malicious strategies are ESSs in this game because neither can be
invaded by the other. The strategy that dominates over time is the one
that starts in the majority. If honest validators start the network,
the honest strategy would be incumbent and the system will remain in
ESS. Similarly, if malicious validators start the network, the network
is malicious. The malicious strategy would be the incumbent strategy
and will remain in ESS. Hence, the population state of the genesis of
the network plays a crucial role for PoS blockchains.
\end{proof}
\section{Experimental Evaluation}
\label{simulations}

%
We performed experimental evaluation to confirm whether penalties are important and what percentage of malicious population can PoS system tolerate. 
Our methodology was to learn how the population state proportions evolve over the generations, the block rounds, using the GameBug
software~\cite{wyttenbach2011gamebug}. The defaults for the relation
between variables and the baseline value of the validator, used in our experimental
simulations, are given in Table~~\ref{tab:simulation-values}. If the
expense is $x$, the reward is taken as $10x$. The default Byzantine
benefit and penalty are choosen as $100x$ for a validator with a
\textit{baseline} balance of $1000x$.
\begin{table}[h]
\centering
\caption{Simulation variables and relative values}
\label{tab:simulation-values}
\begin{tabular}{|l|l|l|}
\hline
\textbf{Variable} & \textbf{Symbol} & \textbf{Value}  \\ \hline\hline
expense               & \textit{e}  & x               \\ \hline
reward                 & $r$& 10x             \\ \hline
Bynzantine benefit      & $b$          & 100x            \\ \hline
penalty                 & $p$ &100x            \\ \hline
\end{tabular}  
\vspace{-1em}
\end{table}

We ran experiments for the reward matrix of the universal reward case (Table~\ref{tab:payoff}) and the reward for work case (Table~\ref{tab:payoff2}) 
with the assumption of more than 90\% of the
validators are honest. We observe that the malicious strategy takes
over when everyone on the network is rewarded (see
Fig.~\ref{fig:case1-9010}). Blue and red signify the proportions of
validators with honest and malicious strategies, respectively. The
X-axis tracks the population proportions, and the Y-axis tracks the
generations, i.e., block rounds. In the no-penalty case, the honest
strategy is an ESS, unaffected by invading mutants (see
Fig.~\ref{fig:case2-9010}). However, we do not have honest ESS with malicious population above 10\%.
\begin{figure}
\centering
\includegraphics[width=0.5\textwidth]{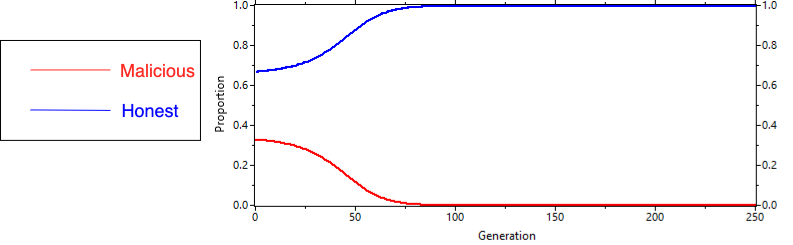}
  \caption{With penalties, honest strategy is ESS, can tolerate one-third malicious}\label{fig:case3-simulation}
  \vspace{-2em}
\end{figure}

We also ran simulations by relaxing the initial proportions at the quorum size of consensus. In other words, the initial proportions of honest and malicious validators are two-thirds and one-third, respectively. For the reward matrix in the reward for work
case (Table~\ref{tab:payoff2}), the proportion of honest drops
significantly over the generations, within 75 block rounds (see
Fig.~\ref{fig:case2-simulation}). Though the honest strategy is ESS,
it takes longer to reach ESS with a high initial proportion of malicious validators. Over generations, the validators would favor a malicious strategy because of high Byzantine benefit if the network starts with
a one-third proportion of malicious validators. For the reward matrix
in the penalty case (Table~\ref{tab:payoff3}), the honest strategy is
an ESS since validators do not shift from the incumbent strategy (see
Fig.~\ref{fig:case3-simulation}). High penalties play a crucial role
in this behaviour, we run a few experiments to verify the same. We
observe that the higher the penalties, the sooner the entire network
becomes honest. We increased the penalties to be 50\% and 100\% of the the
baseline to plot Fig.~\ref{fig:c3500p} and
Fig.~\ref{fig:c31000p}, respectively. We observe that with higher penalties, we reach the state of only honest population quickly.
%
For the penalty case, we tried a 51\% honest majority with default
values, we do not observe an honest ESS. The network cannot tolerate
51\% malicious behaviour (see Fig.~\ref{fig:c3_5149}).

%

\section{Conclusions}
\label{conclusion}
We formulated the block validation game and designed rewards to address free riders and nothing at stake challenges. Using evolutionary game theory, we proved the importance of penalties in maintaining the integrity of the ledger. In the future, we intend to extend this work to asynchronous environments and also analyze how the game adapts when the delegation of stake is allowed. 
\bibliographystyle{IEEEtran}
\bibliography{references}

\begin{thebibliography}{10}
\providecommand{\url}[1]{#1}
\csname url@samestyle\endcsname
\providecommand{\newblock}{\relax}
\providecommand{\bibinfo}[2]{#2}
\providecommand{\BIBentrySTDinterwordspacing}{\spaceskip=0pt\relax}
\providecommand{\BIBentryALTinterwordstretchfactor}{4}
\providecommand{\BIBentryALTinterwordspacing}{\spaceskip=\fontdimen2\font plus
\BIBentryALTinterwordstretchfactor\fontdimen3\font minus
  \fontdimen4\font\relax}
\providecommand{\BIBforeignlanguage}[2]{{%
\expandafter\ifx\csname l@#1\endcsname\relax
\typeout{** WARNING: IEEEtran.bst: No hyphenation pattern has been}%
\typeout{** loaded for the language `#1'. Using the pattern for}%
\typeout{** the default language instead.}%
\else
\language=\csname l@#1\endcsname
\fi
#2}}
\providecommand{\BIBdecl}{\relax}
\BIBdecl

\bibitem{nakamoto2019bitcoin}
S.~Nakamoto, ``Bitcoin: A peer-to-peer electronic cash system,'' Manubot, Tech.
  Rep., 2019.

\bibitem{pound_2021}
\BIBentryALTinterwordspacing
J.~Pound, ``Bitcoin hits \$1 trillion in market value as cryptocurrency surge
  continues,'' \emph{CNBC}, Feb 2021. [Online]. Available:
  \url{https://www.cnbc.com/2021/02/19/bitcoin-hits-1-trillion-in-market-value-as-cryptocurrency-surge-continues.html}
\BIBentrySTDinterwordspacing

\bibitem{wood2014ethereum}
G.~Wood \emph{et~al.}, ``Ethereum: A secure decentralised generalised
  transaction ledger,'' \emph{Ethereum project yellow paper}, vol. 151, no.
  2014, pp. 1--32, 2014.

\bibitem{adams2021uniswap}
H.~Adams, N.~Zinsmeister, M.~Salem, R.~Keefer, and D.~Robinson, ``Uniswap v3
  core,'' 2021.

\bibitem{mejeris2021blockchain}
J.~Mejeris, E.~Au, Y.~Cai, H.-A. Jacobsen, S.~Motepalli, R.~Sun, A.~Veneris,
  G.~Zhang, and S.~Zhang, ``Blockchain for v2x: A taxonomy ofdesign use cases
  and system requirements,'' in \emph{2021 3rd Conf. on Blockchain Research and
  Applicat. for Innovative Networks and Services (BRAINS)}, 2021.

\bibitem{mentzer2018impact}
K.~Mentzer and M.~Gough, ``The impact of cryptokitties on the ethereum
  blockchain,'' in \emph{2018 Annual Conf. (47th)}, 2018, p. 191.

\bibitem{kokoris2018omniledger}
E.~Kokoris-Kogias, P.~Jovanovic, L.~Gasser, N.~Gailly, E.~Syta, and B.~Ford,
  ``Omniledger: A secure, scale-out, decentralized ledger via sharding,'' in
  \emph{2018 IEEE Symposium on Security and Privacy (SP)}.\hskip 1em plus 0.5em
  minus 0.4em\relax IEEE, 2018, pp. 583--598.

\bibitem{stoll2019carbon}
C.~Stoll, L.~Klaa{\ss}en, and U.~Gallersd{\"o}rfer, ``The carbon footprint of
  bitcoin,'' \emph{Joule}, vol.~3, no.~7, pp. 1647--1661, 2019.

\bibitem{kiayias2017ouroboros}
A.~Kiayias, A.~Russell, B.~David, and R.~Oliynykov, ``Ouroboros: A provably
  secure proof-of-stake blockchain protocol,'' in \emph{Annual Int. Cryptology
  Conf.}\hskip 1em plus 0.5em minus 0.4em\relax Springer, 2017, pp. 357--388.

\bibitem{kwon2014tendermint}
J.~Kwon, ``Tendermint: Consensus without mining,'' \emph{Draft v. 0.6, fall},
  vol.~1, no.~11, 2014.

\bibitem{gilad2017algorand}
Y.~Gilad, R.~Hemo, S.~Micali, G.~Vlachos, and N.~Zeldovich, ``Algorand: Scaling
  byzantine agreements for cryptocurrencies,'' in \emph{Proceedings of the 26th
  Symposium on Operating Systems Principles}, 2017, pp. 51--68.

\bibitem{algorand-faq}
\BIBentryALTinterwordspacing
{Algorand FAQ}. Algorand Foundation. Accessed: 2021-03-18. [Online]. Available:
  \url{https://algorand.foundation/faq}
\BIBentrySTDinterwordspacing

\bibitem{fooladgar2020incentive}
M.~Fooladgar, M.~H. Manshaei, M.~Jadliwala, and M.~A. Rahman, ``On incentive
  compatible role-based reward distribution in algorand,'' in \emph{2020 50th
  Annual IEEE/IFIP Int. Conf. on Dependable Systems and Networks (DSN)}.\hskip
  1em plus 0.5em minus 0.4em\relax IEEE, 2020, pp. 452--463.

\bibitem{aval-token}
\BIBentryALTinterwordspacing
S.~Buttolph, A.~Moin, K.~Sekniqi, and E.~G. Sirer. (2020) {Avalanche Native
  Token (\$AVAX) Dynamics}. Ava Labs. Accessed: 2021-03-18. [Online].
  Available: \url{https://files.avalabs.org/papers/token.pdf}
\BIBentrySTDinterwordspacing

\bibitem{david2018ouroboros}
B.~David, P.~Ga{\v{z}}i, A.~Kiayias, and A.~Russell, ``Ouroboros praos: An
  adaptively-secure, semi-synchronous proof-of-stake blockchain,'' in
  \emph{Annual Int. Conf. on the Theory and Applications of Cryptographic
  Techniques}.\hskip 1em plus 0.5em minus 0.4em\relax Springer, 2018, pp.
  66--98.

\bibitem{cosmos}
\BIBentryALTinterwordspacing
C.~Unchained. (2018) {Cosmos Validator Economics — Bridging the Economic
  System of Old into the New Age of Blockchains}. Cosmos Blog. Accessed:
  2021-03-18. [Online]. Available:
  \url{"https://blog.cosmos.network/economics-of-proof-of-stake-bridging-the-economic-system-of-old-into-the-new-age-of-blockchains-3f17824e91db"}
\BIBentrySTDinterwordspacing

\bibitem{eth2stake}
\BIBentryALTinterwordspacing
J.~Beck. (2020) {Rewards and Penalties on Ethereum 2.0 [Phase 0]}. ConsenSys.
  Accessed: 2021-03-18. [Online]. Available:
  \url{https://consensys.net/blog/codefi/rewards-and-penalties-on-ethereum-20-phase-0/}
\BIBentrySTDinterwordspacing

\bibitem{eth2pos}
\BIBentryALTinterwordspacing
{Proof of Stake FAQs}. Ethereum Wiki. Accessed: 2021-03-18. [Online].
  Available: \url{https://eth.wiki/en/concepts/proof-of-stake-faqs}
\BIBentrySTDinterwordspacing

\bibitem{polkadot}
\BIBentryALTinterwordspacing
{Staking}. Polkadot wiki. Accessed: 2021-03-18. [Online]. Available:
  \url{https://wiki.polkadot.network/docs/en/learn-staking}
\BIBentrySTDinterwordspacing

\bibitem{chacko2021my}
J.~A. Chacko, R.~Mayer, and H.-A. Jacobsen, ``Why do my blockchain transactions
  fail? a study of hyperledger fabric,'' in \emph{Proceedings of the 2021 Int.
  Conf. on Management of Data}, 2021, pp. 221--234.

\bibitem{zhang2021prosecutor}
G.~Zhang and H.-A. Jacobsen, ``{Prosecutor: An Efficient BFT Consensus
  Algorithm with Behavior-aware Penalization Against Byzantine Attacks},'' in
  \emph{Middleware '21: 22st ACM/IFIP Int. Middleware Conf.}, 2021.

\bibitem{smith1973logic}
J.~M. Smith and G.~R. Price, ``The logic of animal conflict,'' \emph{Nature},
  vol. 246, no. 5427, pp. 15--18, 1973.

\bibitem{vincent2005evolutionary}
T.~L. Vincent and J.~S. Brown, \emph{Evolutionary game theory, natural
  selection, and Darwinian dynamics}.\hskip 1em plus 0.5em minus 0.4em\relax
  Cambridge University Press, 2005.

\bibitem{friedman1998economic}
D.~Friedman, ``On economic applications of evolutionary game theory,''
  \emph{Journal of evolutionary economics}, vol.~8, no.~1, pp. 15--43, 1998.

\bibitem{witt2016specific}
U.~Witt, ``What is specific about evolutionary economics?'' in \emph{Rethinking
  Economic Evolution}.\hskip 1em plus 0.5em minus 0.4em\relax Edward Elgar
  Publishing, 2016.

\bibitem{mailath1998people}
G.~J. Mailath, ``Do people play nash equilibrium? lessons from evolutionary
  game theory,'' \emph{Journal of Economic Literature}, vol.~36, no.~3, pp.
  1347--1374, 1998.

\bibitem{buss2015evolutionary}
D.~Buss, \emph{Evolutionary psychology: The new science of the mind}.\hskip 1em
  plus 0.5em minus 0.4em\relax Psychology Press, 2015.

\bibitem{laland2011sense}
K.~N. Laland, G.~Brown, and G.~R. Brown, \emph{Sense and nonsense: Evolutionary
  perspectives on human behaviour}.\hskip 1em plus 0.5em minus 0.4em\relax
  Oxford University Press, 2011.

\bibitem{tooby2005conceptual}
J.~Tooby and L.~Cosmides, ``Conceptual foundations of evolutionary
  psychology.'' 2005.

\bibitem{easley2010networks}
D.~Easley, J.~Kleinberg \emph{et~al.}, \emph{Networks, crowds, and
  markets}.\hskip 1em plus 0.5em minus 0.4em\relax Cambridge university press
  Cambridge, 2010, vol.~8.

\bibitem{smith1976evolution}
J.~M. Smith \emph{et~al.}, ``Evolution and the theory of games,''
  \emph{American scientist}, vol.~64, no.~1, pp. 41--45, 1976.

\bibitem{cowden2012game}
C.~Cowden, ``Game theory, evolutionary stable strategies and the evolution of
  biological interactions,'' \emph{Nature Education Knowledge}, vol.~3, no.~6,
  2012.

\bibitem{liu2019survey}
Z.~Liu, N.~C. Luong, W.~Wang, D.~Niyato, P.~Wang, Y.-C. Liang, and D.~I. Kim,
  ``A survey on applications of game theory in blockchain,'' \emph{arXiv
  preprint arXiv:1902.10865}, 2019.

\bibitem{zhang2020analysing}
S.~Zhang, K.~Zhang, and B.~Kemme, ``Analysing the benefit of selfish mining
  with multiple players,'' in \emph{2020 IEEE Int. Conf. on Blockchain
  (Blockchain)}.\hskip 1em plus 0.5em minus 0.4em\relax IEEE, 2020, pp. 36--44.

\bibitem{kroll2013economics}
J.~A. Kroll, I.~C. Davey, and E.~W. Felten, ``The economics of bitcoin mining,
  or bitcoin in the presence of adversaries,'' in \emph{Proceedings of WEIS},
  vol. 2013, 2013, p.~11.

\bibitem{liu2018evolutionary}
X.~Liu, W.~Wang, D.~Niyato, N.~Zhao, and P.~Wang, ``Evolutionary game for
  mining pool selection in blockchain networks,'' \emph{IEEE Wireless
  Communications Letters}, vol.~7, no.~5, pp. 760--763, 2018.

\bibitem{ni2019evolutionary}
Z.~Ni, W.~Wang, D.~I. Kim, P.~Wang, and D.~Niyato, ``Evolutionary game for
  consensus provision in permissionless blockchain networks with shards,'' in
  \emph{ICC 2019-2019 IEEE Int. Conf. on Communications (ICC)}.\hskip 1em plus
  0.5em minus 0.4em\relax IEEE, 2019, pp. 1--6.

\bibitem{kim2019mining}
S.~Kim and S.-G. Hahn, ``Mining pool manipulation in blockchain network over
  evolutionary block withholding attack,'' \emph{IEEE Access}, vol.~7, pp.
  144\,230--144\,244, 2019.

\bibitem{iyer2018crypto}
K.~Iyer and C.~Dannen, ``Crypto-economics and game theory,'' in \emph{Building
  Games with Ethereum Smart Contracts}.\hskip 1em plus 0.5em minus 0.4em\relax
  Springer, 2018, pp. 129--141.

\bibitem{chang2020incentive}
Z.~Chang, W.~Guo, X.~Guo, Z.~Zhou, and T.~Ristaniemi, ``Incentive mechanism for
  edge-computing-based blockchain,'' \emph{IEEE Transactions on Industrial
  Informatics}, vol.~16, no.~11, pp. 7105--7114, 2020.

\bibitem{chiu2019incentive}
J.~Chiu and T.~Koeppl, ``Incentive compatibility on the blockchain,'' in
  \emph{Social Design}.\hskip 1em plus 0.5em minus 0.4em\relax Springer, 2019,
  pp. 323--335.

\bibitem{xuan2020incentive}
S.~Xuan, L.~Zheng, I.~Chung, W.~Wang, D.~Man, X.~Du, W.~Yang, and M.~Guizani,
  ``An incentive mechanism for data sharing based on blockchain with smart
  contracts,'' \emph{Computers \& Electrical Engineering}, vol.~83, p. 106587,
  2020.

\bibitem{li2017novel}
J.~Li, G.~Liang, and T.~Liu, ``A novel multi-link integrated factor algorithm
  considering node trust degree for blockchain-based communication.''
  \emph{KSII Transactions on Internet \& Inform. Systems}, vol.~11, no.~8,
  2017.

\bibitem{sompolinsky2015secure}
Y.~Sompolinsky and A.~Zohar, ``Secure high-rate transaction processing in
  bitcoin,'' in \emph{Int. Conf. on Financial Cryptography and Data
  Security}.\hskip 1em plus 0.5em minus 0.4em\relax Springer, 2015, pp.
  507--527.

\bibitem{tversky1979prospect}
A.~Tversky and D.~Kahneman, ``Prospect theory: An analysis of decision under
  risk,'' \emph{Econometrica}, vol.~47, no.~2, pp. 263--291, 1979.

\bibitem{tversky1991loss}
------, ``Loss aversion in riskless choice: A reference-dependent model,''
  \emph{The quarterly journal of economics}, vol. 106, no.~4, pp. 1039--1061,
  1991.

\bibitem{pfattheicher2016misperceiving}
S.~Pfattheicher and S.~Schindler, ``Misperceiving bullshit as profound is
  associated with favorable views of cruz, rubio, trump and conservatism,''
  \emph{PloS one}, vol.~11, no.~4, p. e0153419, 2016.

\bibitem{gachter2009experimental}
S.~G{\"a}chter, H.~Orzen, E.~Renner, and C.~Starmer, ``Are experimental
  economists prone to framing effects? a natural field experiment,''
  \emph{Journal of Economic Behavior \& Organization}, vol.~70, no.~3, pp.
  443--446, 2009.

\bibitem{smith1982evolution}
J.~M. Smith, \emph{Evolution and the Theory of Games}.\hskip 1em plus 0.5em
  minus 0.4em\relax Cambridge university press, 1982.

\bibitem{wyttenbach2011gamebug}
\BIBentryALTinterwordspacing
R.~Wyttenbach, ``Gamebug software,'' 2011. [Online]. Available:
  \url{http://hoylab.cornell.edu/download.html}
\BIBentrySTDinterwordspacing

\end{thebibliography}

\end{document}